\documentclass[submission,copyright,creativecommons]{eptcs}
\providecommand{\event}{FESCA 2015} 

\newcommand{\cat}{{}^\frown{}}
\newcommand{\trace}[1]{\langle #1 \rangle}

\usepackage{mathpartir} 

\usepackage{listings}

\usepackage[utf8x]{inputenc}  

\usepackage{amsmath}
\usepackage{amsthm}
\newtheorem{theorem}{Theorem}

\newtheorem{corollary}[theorem]{Corollary}

\usepackage{float}
\floatstyle{boxed}
\restylefloat{figure}

\usepackage[usenames]{xcolor}

\usepackage{stmaryrd} 
\usepackage{mathtools} 

\title{A Denotational Semantics for \\Communicating Unstructured Code}
\author{Nils Jähnig \qquad\qquad Thomas Göthel \qquad\qquad Sabine Glesner
\institute{Technische Universität Berlin, Germany}
\email{\{nils.jaehnig, thomas.goethel, sabine.glesner\}@tu-berlin.de}
}
\def\titlerunning{Denotational Semantics for CUC}
\def\authorrunning{N. Jähnig, T. Göthel \& S. Glesner}
\begin{document}
\maketitle

\begin{abstract}
An important property of programming language semantics is that they should be compositional.
However, unstructured low-level code contains {\tt goto}-like commands making it hard to define a semantics that is compositional. 
In this paper, we follow the ideas of Saabas and Uustalu~\cite{Saabas05acompositional} to structure low-level code. This gives us the possibility to define a compositional denotational semantics based on least fixed points to allow for the use of inductive verification methods. We capture the semantics of communication using finite traces similar to the denotations of CSP. In addition, we examine properties of this semantics and give an example that demonstrates reasoning about communication and jumps.
With this semantics, we lay the foundations for a proof calculus that captures both, the semantics of unstructured low-level code and communication.
\end{abstract}

\section{Introduction}
\label{sec:introduction}

Electronic devices are increasingly integrated into our daily lives. Most of them are ``invisible'' and would, ideally, run forever once deployed. Often this is not a crucial feature, as a smart phone or a server can be restarted if they crash. But in cases where the system cannot be restarted, the correctness of those systems needs to be ensured. Examples of such cases are medical devices and automotive systems.

By the {\em correctness of a system} we mean a {\em formal conformance relation} between an abstract specification and its implementation. We consider the specification language of Communicating Sequential Processes (CSP)~\cite{Schneider1999} as suitable technique to describe and verify the communication behavior of concurrent programs. The advantage is that properties can be verified on the abstract model of CSP and conformance ensures that these properties are preserved to the implementation level. 
We assume as in our modeling and verification approach presented in~\cite{VATES_Approach} that, having refined the CSP specification sufficiently, it is implemented in a high-level language such as C++ and then compiled to a low-level language such as LLVM. 
Instead of showing conformance between a high-level language and CSP, we show conformance of the corresponding low-level representation instead.
This approach has the advantage that a complex formal semantics of a high-level language can be avoided and tightens the gap between verified code and executed code. 
We choose an LLVM-like unstructured language that we previously presented in~\cite{Bartels:2014aa} because it is platform independent. In order to model communicating programs, we have included a generic instruction that allows for communication between unstructured programs.

In this paper, we give a formal denotational semantics to the implementation level focussing on communicating unstructured code. The particular challenge lies in the fact that denotational semantics for unstructured code are usually not compositional because of the presence of some variant of the {\tt goto} command. This usually requires the whole code available when showing properties.
Furthermore, we consider communication. 
To overcome the first problem, we structure unstructured code following ideas from Saabas and Uustalu~\cite{Saabas05acompositional}. While they define a compositional bigstep semantics for an unstructured language, we define a denotational semantics based on fixpoints, which additionally copes with communication and nontermination. The latter features are close to the approach taken in CSP. As we require our approach to enable automation and also to gain confidence in the correctness of our semantical construction, we formalize all results in the theorem prover Isabelle/HOL\footnote{When the intended Hoare calculus (see Section~\ref{sec:calculus}) is finished, we plan to publish the Isabelle theories in the {\em Archive of Formal Proofs} (\url{http://afp.sourceforge.net/}). If you are interested in the current version of the theories, please contact us.}~\cite{Nipkow2002}. Our denotational semantics gives us the groundwork for developing a compositional proof calculus that not only allows to formulate properties about possible states of unstructured programs but also allows us to reason about the communication behavior using traces.

The rest of this paper is structured as follows. Related work is dicussed in the next section. In Section~\ref{sec:background}, we give necessary background information about structured code, the concept of denotational semantics, CSP, and Isabelle/HOL. In Section~\ref{sec:cuc}, we introduce our low-level language, which we call {\em Communicating Unstructured Code} and for which we define a smallstep operational semantics in Section~\ref{sec:operational}. Our main contribution is described in Section~\ref{sec:denotational} where we introduce our denotational semantics for Communicating Unstructured Code. 
In Section~\ref{sec:analysis}, we analyze the denotational semantics and prove conformance to the intuitive operational semantics. 
In Section \ref{sec:calculus}, we define a rule to ease argumentation about the semantics. This gives an idea what a future calculus will look like. We illustrate the utility of our denotational semantics in Section~\ref{sec:example}. We give a conclusion and pointers to future work in Section~\ref{sec:conclusion}.

\section{Related Work}
There are some attempts to give unstructured code a semantics for later verification. Tews~\cite{Tews:2004aa} developed a compositional semantics for a C-like language with {\tt goto}, which is used to verify Duff's device. However, this approach does not model communication.

Saabas and Uustalu~\cite{Saabas05acompositional} present a compositional bigstep semantics of an unstructured language. To this end, a generic structuring mechanism for the code is presented, which makes the semantics compositional and also allows for a compositional proof calculus. Although they formally relate a high-level and a low-level language, they do not relate a process specification with the low-level language. Communication is not considered.

A denotational semantics and a proof calculus for a high-level language with communication were already defined by Zwiers~\cite{zwiers1989compositionality}. As low-level code is not considered, the semantics is not directly applicable, but we will follow this approach when designing a proof calculus in future work.
A preliminary rule of this calculus is presented in this paper.

Our unstructured language that we introduce in Section~\ref{sec:cuc} is based on our previous work~\cite{Bartels:2014aa} enhanced with communication capabilities. The approach in~\cite{Bartels:2014aa} focuses on a smallstep and a bigstep operational semantics based on which a compositional proof calculus is built. We used similar semantics in~\cite{BarGle2011} to show correspondence between unstructured code and (Timed) CSP processes.
In~\cite{BarGle2011} we used events as observation points, but did not consider actual communication.
In addition to considering also communication, we aim in future at relating unstructured code and CSP processes based on our denotational semantics, which facilitates refinement-based verification instead of bisimulation equivalence as used in~\cite{BarGle2011}.

\section{Background}
\label{sec:background}

In this section, we introduce the concept of structuring unstructered code and give brief introductions to denotational semantics and Communication Sequential Processes.

\subsection{Structured Code}
\label{ssec:structuredcode}

It is difficult to obtain a compositional semantics for unstructured low-level code, as compositionality usually involves splitting the code into parts according to its structure.
The simplest structure is a list of instructions, which can be split into head and tail. 
But also in unstructured code there are blocks of instructions that belong to the same functionality, and it is desirable to reason about such blocks as components and not instruction after instruction appended to the previous code (as it would be the case for lists). Therefore we choose a structure which allows for those blocks to be reasoned about as a component, namely trees.
To this end, we follow a similar approach to the one presented in~\cite{Saabas05acompositional} by Saabas and Uustalu.

In their approach, every instruction has a unique label (e.g., a natural number) that refers to it, and the semantics depends on an explicit {\em program counter} that holds the label of the next instruction to execute. After executing one instruction, the program counter is usually incremented by one so that it points to the next instruction. Only {\em control flow instructions}, e.g., {\tt goto} can set the program counter to a different value. This way, conditionals or loops can be constructed.

An {\em unstructured program} can be intuitively formalized as a set of {\em labeled instructions}, i.e., pairs of labels and instructions. All labels are assumed to be unique. The shortcoming of this formalization is that there is no useful structure on which induction can be applied. E.g., as the whole code needs to be considered all the time due to the {\tt goto} instructions, the operational smallstep semantics on it (see Section~\ref{sec:operational}) is not compositional\footnote{Compositionality here refers to ``sequential'' compsitionality. In this paper we do not consider ``parallel'' compositionality.} with respect to the program. So a denotational semantics on unstructured code will also not be compositional.

To overcome this problem, we introduce a tree structure on the instructions.
Let {\em code} adhere to the following grammar
\[
code \Coloneqq (label:: instruction) \mid code \oplus code
\] 
where $\oplus$ is the sequential composition operator. As described in the next subsection, we can use this structure to reason about the single components (i.e., subtrees),  and compose their meaning, thus obtain compositionality. 
To make sure that the structure does not influence the semantics of a program, all possibilities to structure a given set of instructions must have the same meaning (for our semantics this is a corollary of the conformance proved in Section~\ref{ssec:conformance}). We can choose which structure on the code fits our needs best. Usually, this is a structure where as many jumps as possible are inside a component. By this, the entry and exit possibilities for each component are minimized. 
This will help in future work when choosing pre and post conditions, as for known code, not every instruction label needs to be considered as a possible entry point.

It is important to note that even though we call $\oplus$ the {\em sequential composition} (opposed to a possible {\em parallel composition}), it is not the sequential composition known from structured programming languages.
The important difference is that instructions combined with $\oplus$ can create loops, if branch instructions are contained.
As loops are only dealt with in rules about $\oplus$, we call $\oplus$ the looping construct of the language.

To be able to relate structured and unstructured code, a projection $\cup_{sc}$ from {\em structured code} to {\em labeled instruction sets} is defined, which ``forgets'' the structure. 
Observe, that the operation $\oplus$ is only defined for structured code. 

Having introduced the concept of structured code in our setting, we are able to define a compositional denotational semantics of our low-level language. It is particularly amenable to fixpoint induction, which we describe subsequently.

\subsection{Denotational Semantics}
\label{ssec:denotationalSemantics}

In contrast to operational semantics, which describes how single instructions manipulate the state, in denotational semantics each program is assigned a function (its {\em denotation}) which maps an initial state to the meaning of the program. This can be, e.g., a final state or a set of traces as in the case of CSP. 
As denotational semantics can be characterized via operational semantics  and operational semantics can be defined compositionally, boundaries between the two can be vague.

To be able to assign a denotation to programs containing loops, a function $F$ from {\em denotation} to {\em denotation} is considered, which extends the denotation with the semantics for one execution of the loop. The least fixpoint of this function $F$ is then the denotation of the program containing the loop.

To define a least fixpoint, we use a chain-complete partial order on the denotations.
This is usually obtained by lifting the chain-complete partial order on the states, i.e., co-domain of the denotations.
In general, the resulting denotation may not be computable in finite time, but fixpoint induction can be used to reason about properties of denotations of programs containing loops. 
Fixpoint induction (or {\em Scott induction}, e.g., in~\cite{Winskel:1993:FSP:151145}) requires the function $F$ to be continuous, and the property of interest $I$ to be admissible, i.e., if $I$ holds for all elements of a chain, then $I$ also holds for the supremum of the chain. Finally, for every denotation $f$, if $I(f)$ holds, then $I(F(f))$ has to hold as well (step case).

\subsection{Communicating Sequential Processes}
\label{ssec:csp}

{\em Note: CSP is not necessary for the following definitions, but is helpful as the denotational semantics is designed with CSP in mind.}

Communicating Sequential Processes (CSP)~\cite{Schneider1999} is a process algebra, originally introduced in~\cite{Hoare:1978:CSP:359576.359585}.
It is designed specifically to model concurrent processes, which communicate via events. Communication is synchronous and thus can be used to synchronize processes or exchange data.

Processes can be constructed from the basic processes STOP and SKIP and using operators such as event prefixing, external and internal choice, interrupts, and sequential and parallel composition. 

CSP is suited for modeling various layers of abstraction as described in~\cite{Schneider1999}: Specification, design and implementation, where specification is most abstract, and implementation can be directly converted into program code. Abstraction levels can be mixed in a process, and CSP processes across all abstraction levels can be put in relation via {\em refinements}. Informally a process $Q$ refines a process $P$, if the behavior of process $Q$ is a subset of the behavior of process $P$.

There are two important semantic models for CSP with raising complexity and expressiveness: {\em 1)} The trace semantics, which describes the communication histories of processes, which we use in this paper. {\em 2)} The stable failures semantics, which additionally captures the events a process can refuse after a trace.
In this paper, we focus on traces, but we plan to extend it to failures in future work.

The automatic refinement checker FDR3~\cite{DBLP:conf/tacas/Gibson-RobinsonABR14} supports refinement checks for both mentioned semantics.

\section{Communicating Unstructured Code}
\label{sec:cuc}

In this section, we introduce {\em Communicating Unstructured Code} (CUC), which we adapt from our previous work~\cite{Bartels:2014aa} and enhance it with a communication primitive.

On the one hand, we want to be as close to low-level code as possible to reduce the gap between executed code and verified code.
On the other hand, we want to focus on communication and not its implementation.
Therefore, we decided to use an unstructured language with a higher level construct for communication. 

We keep our language generic and simple. Though being simple, it is powerful (as it contains conditional branches). This simplicity allows for manageable semantics design and proofs.
Our language consists of three basic instructions: 
\[
instruction \coloneqq \texttt{do f} \mid \texttt{cbr b m n} \mid \texttt{comm ef f}
\]
In the following, we give an informal explanation of these instructions.

\textbf{\texttt{do f}} -- The command \texttt{do} is a generalized assignment. \texttt{f} is a function from {\em state} to {\em set of states} and it is applied  to the current state. The resulting state is one element of the set returned by {\tt f}. The instruction can thus be thought of as a {\em nondeterministic multiple assignment}.

\texttt{cbr b m n} -- The instruction {\tt cbr} is a usual conditional branch. If the function \texttt{b} from {\em state} to {\em bool} evaluates to {\em True} then the {\em program counter} is set to \texttt{m} else to \texttt{n}.

\texttt{comm ef f} --  The command {\tt comm} is the communication primitive. It communicates an event from the result of \texttt{ef} and then changes the state according to \texttt{f}, where \texttt{ef} is a function from {\em state} to {\em set of events} and \texttt{f} is a function from {\em state} and {\em event} to {\em state}. Observe that here \texttt{f} is deterministic to ease reasoning. We reserve nondeterminism of the successor state to the instruction \texttt{do f}.

The instructions \texttt{do f} and \texttt{comm ef f} are abstract instructions (instructions schemes), which can be instantiated. For example, the instance on the left models an assignment of $y$ to $x$. The instance on the right adds $x$ to $y$ and stores the result to $z$.

\begin{center}
\begin{tabular}{c p{10mm} c}
``$x \coloneqq y$'' & &
``$z \coloneqq x + y$''
\\
$\texttt{do } (\lambda \sigma.\; \sigma[x \leftarrow \sigma(y)])$ & &
$\texttt{do } (\lambda \sigma.\; \sigma[z \leftarrow \sigma(x)+\sigma(y)])$
\end{tabular} 
\end{center}

\noindent The next instance defines an input of values of type $T$ over channel $in$ storing the value in variable $x$.

\begin{center}
``$x \coloneqq \text{input}(in::T)$''

$\texttt{comm } (\lambda \sigma.\; \{in.v \mid v \in T\})(\lambda \sigma\; ev .\; \sigma[x \leftarrow val(ev)])$
\end{center}

\noindent In the last example and in Section~\ref{sec:example}, events are of the type $channels \times values$ (written in the dot-notation $channel.value$ as usual in CSP). This is just a convenient instantiation, as there are no strucutral requirements on the type of events (also as in CSP). Apart from the chosen type of events, there is no concept of channels in CUC.

In this section, we have presented the instructions of our language and have given its informal semantics. 
In the next section, we first define the more intuitive smallstep semantics to capture our ideas, and then define the denotational semantics in Section~\ref{sec:denotational}, as it fits better our ultimate goal to relate CUC programs and CSP processes. 
In Section~\ref{ssec:conformance}, we show conformance between the smallstep and the denotational semantics, to ensure that the denotational semantics adheres to our ideas.

\subsection{Operational Smallstep Semantics}
\label{sec:operational}

As we consider unstructured code, a program is given as a set of {\em labeled instructions}, i.e., pairs of identifier (label) and instructions. We assume that the identifiers are unique.

We model the state as consisting of three parts: The usual state (i.e., the values of variables), the program counter, and the trace (i.e., the communication history as a list of events). 
Observe that each instruction roughly corresponds to one part of the state: \texttt{do} changes the variables,  \texttt{cbr} the program counter, and \texttt{comm} the trace. The exceptions are that both, \texttt{do} and \texttt{comm}, increase the program counter by one, and that \texttt{comm} can also change the state to store communicated values.
See Figure~\ref{fig:smallstep} for the  smallstep semantics, which captures the behavior described in Section~\ref{sec:cuc}.
Each instruction can only be applied if the $pc$ is pointing to its label. We explain the rules in detail:

{\sc S-do} -- Only the variables are changed, possibly nondeterministically, according to $f$. Traces stay the same and the program counter is increased by one.

{\sc S-cbr} --  There are two rules, one for the case when the condition $b$ evaluates to True, and one when it evaluates to False. In both cases only the $pc$ is changed accordingly, and variables and trace remain the same.

{\sc S-comm} -- $ef$ determines which events are offered for communication. The communicated event is appended to the existing trace. The successor state is calculated according to $f$, this time deterministically and additionally depending on the communicated event. $f$ is deterministic to focus on communication. Again, the program counter is increased by one.

\newcommand{\smallstep}{\xrightarrow{code}}
\begin{figure}[h]
\[
\sigma \xrightarrow{code} \vartheta \coloneqq (\sigma, \vartheta) \in smallstep(code) 
\]

\[
\inferrule*[right=S-do]
{(\ell :: \texttt{do }f) \in code \\
t \in f(s)}
{(tr, s, \ell) \xrightarrow{code} (tr, t, \ell + 1)}
\]

\[
\inferrule*[right=S-cbr True]
{(\ell :: \texttt{cbr }b\;m\;n) \in code \\
b(s) = \text{True}
}
{(tr, s, \ell) \xrightarrow{code} (tr, s, m)}
\]

\[
\inferrule*[right=S-cbr False]
{(\ell:: \texttt{cbr }b\;m\;n) \in code \\
b(s) = \text{False}
}
{(tr, s, \ell)\xrightarrow{code} (tr, s, n)}
\]

\[
\inferrule*[right=S-comm]
{(\ell:: \texttt{comm }ef\;f) \in code \\ 
ev \in ef(s)\\
tr_t = tr_s\cat ev\\
t = f(s, ev)
}
{(tr_s, s, \ell) \xrightarrow{code} (tr_t, t, \ell+1)}
\]
\caption{Smallstep Semantics for CUC}
\label{fig:smallstep}
\end{figure}

The smallstep semantics defines the executions of single instructions.
To relate the operational semantics with the denotational semantics we need the {\em multistep} relation, which is the reflexive transitive closure of the smallstep relation, and defines the semantics for executions of a set of instructions.
See the multistep semantics in Figure~\ref{fig:multistep}.%
\newcommand{\multistep}[1]{\xrightarrow{#1}\hspace{-4pt}{}^*}
\begin{figure}
\[
\sigma \multistep{code} \vartheta \coloneqq (\sigma, \vartheta) \in multistep(code) 
\]

\[
\inferrule*[right=M-base]
{ }
{(tr, s, \ell ) \multistep{code} (tr, s, \ell)}
\]

\[
\inferrule*[right=M-step]
{(tr_s, s, \ell_s) \smallstep (tr_{u}, u, \ell_{u})\\
(tr_{u}, u, \ell_{u}) \multistep{code} (tr_t, t, \ell_t)}
{(tr_s, s, \ell_s) \multistep{code} (tr_t, t, \ell_t)}
\]
\caption{Multistep Semantics for CUC}
\label{fig:multistep}
\end{figure}
The first rule ({\sc M-base}) is the base case, stating that every step is reachable from itself (reflexivity); the second rule ({\sc M-step}) is the step rule, stating that if a state is reachable, its smallstep successors are reachable, too (transitivity).
Please note that the operational semantics is not defined for $\oplus$, as the operational semantics is defined on unstructured code.

\subsection{Denotational Semantics}
\label{sec:denotational}

We presented previously a compositional bigstep semantics for a low-level language {\em without} communication in~\cite{Bartels:2014aa}. In contrast, we require our extended semantics {\em with} communication to be close to the denotational CSP semantics. To this end, we define a denotational semantics which captures the state information but also the finite traces of program executions.

See Figure~\ref{fig:denotational} for the denotational semantics defined on the structured variant of CUC. 
The semantic function $\llbracket \cdot \rrbracket$ maps {\em code} to its {\em denotation}.
We use the $\mu$ operator to denote a {\em least fixpoint} and the underlying chain-complete partial order is the point-wise subset relation \mbox{($f \leq g \coloneqq \forall S. \; f(S) \subseteq g(S)$} where $S$ is a set of states).

\begin{figure}

\begin{alignat*}{2}
    \textsc{D-do} &&&\\
    && \llbracket\ell:: \texttt{do }f\rrbracket(S) & \coloneqq
    S \cup \{(tr, t, \ell + 1) \mid (tr, s, \ell) \in S \wedge t \in f(s)  \} \\
\\
    \textsc{D-cbr} &&&\\
    && \llbracket\ell :: \texttt{cbr }b\;m\;n\rrbracket(S) & \coloneqq S \cup \{(tr, s, pc) \mid (tr, s, \ell) \in S \wedge ( b(s) \wedge pc=m \vee \neg b(s) \wedge pc=n)  \} \\
\\
    \textsc{D-comm} &&&\\
    && \llbracket\ell:: \texttt{comm }ef\; f\rrbracket(S) & \coloneqq S \cup \{(tr\cat ev, t, \ell + 1) \mid 
    (tr, s, \ell) \in S \wedge ev \in ef(s) \wedge t = f(s, ev)  \}
\end{alignat*}
\begin{alignat*}{2}
    \textsc{D-seq} &&&\\
    && \llbracket code_1 \oplus code_2\rrbracket & \coloneqq \big(\mu d \hspace{-0.4mm}.\; extend(code_1, code_2) \; (d)\big)
\\
    \textsc{D-ext} &&&\\
    && extend(code_1, code_2)\;  (d) & \coloneqq
    \lambda S. \; S \cup d\big(\llbracket code_1\rrbracket(S)\big) \cup d\big(\llbracket code_2\rrbracket(S)\big)
\end{alignat*}

\caption{Denotational Semantics for CUC}
\label{fig:denotational}
\end{figure}

The smallstep semantics operates on single states, but we use sets of states in the denotational semantics to capture nondeterminism. Furthermore, we require sets also as input of the denotations so that we are able to chain functions directly. This facilitates sequential composition, as introduced in the background. 
Finally, we can keep track of all intermediate states and the corresponding traces, by simply extending the set input to the denotation (``$S \; \cup$'' on the right-hand side of function definitions). 
This construction is also used in CSP and enables {\em prefix closure} (see Section~\ref{ssec:prefixclosure}).

The single step rules {\sc D-do}, {\sc D-cbr} and {\sc D-comm} are very similar to their operational semantics counterparts. For all states where the program points to the label $\ell$ of the instruction, all successor states are added to the resulting set.

$extend$ ({\sc D-ext}) extends a given denotation $d$ with the executions of the two components $code_1$ and $code_2$ respectively. ($extend$ corresponds to the function $F$ in Section~\ref{ssec:denotationalSemantics}.) With its help we describe in {\sc D-seq} the sequential composition $\oplus$ which is modeled as fixpoint of $extend$ and thus as the repeated, possibly alternating application of the two components. Remember, that $\oplus$ is also the looping construct (see Section~\ref{ssec:structuredcode}).

We have presented the operational and denotational semantics and show important properties of the latter in the next section.

\section{Analysis of the Denotational Semantics}
\label{sec:analysis}

Our longterm goal is to define a refinement proof calculus, which relates CUC programs and CSP processes. This can only succeed, if we can ensure that the CUC semantics is close enough to the CSP semantics. Therefore we show a property related to structured code and two selected properties of CSP for the denotational semantics: {\em i)} Compositionality, which is also important for the future construction of a compositional calculus, and {\em ii)} prefix closure, which will be important when showing that a CUC program refines a CSP process considering failures (which allow for the verification of liveness properties) in future work.
We start outlining the proof that the denotational semantics conforms to the operational multistep semantics, to ensure that the ideas formulated in the smallstep semantics are also captured by the denotational semantics.

\subsection{Conformance Proof}
\label{ssec:conformance}

By conformance we mean that, starting in an arbitrary state, the reachable states through the denotational semantics and the multistep semantics of a given structured code and its unstructured projection respectively are the same: 

\[
(tr_t, t, pc_t)\in \llbracket code\rrbracket(\{(tr_s, s, pc_s)\})
\Longleftrightarrow
(tr_s, s, pc_s) \multistep{\cup_{sc}code} (tr_t, t, pc_t) 
\]

We outline the two directions of the conformance proof. We have formalized full proofs in Isabelle/HOL. First, we argue that all states in the returned set of a denotation are reachable through the multistep relation, and then the inverse direction. As we use the structuring technique for unstructured code from Saabas and Uustalu~\cite{Saabas05acompositional}, we follow their proofs adapted for our denotational semantics.
In both directions we use the similarity of single steps of the two semantics. The interesting case in each direction is the sequential composition case.

\paragraph{Denotational Semantics $\Longrightarrow$ Multistep}
We perform induction over the structure of {\em code}.
The induction hypothesis states that for the subcomponents the denotational semantics implies the multistep semantics. To show the sequential composition case, we combine the hypotheses for the subcomponents using fixpoint induction.
 
\paragraph{Multistep $\Longrightarrow$ Denotational Semantics}
The main problem is that the multistep relation operates on instruction sets, thus is insensible to the structure of {\em code}. We need to extend the {\em multistep} relation with a step counter $k$, in order to invoke induction over the number of steps.  

We still perform induction over the structure of $code$, but within each case we use case-distinction or induction over $k$, respectively. 

An important insight from Saabas and Uustalu~\cite{Saabas05acompositional} for the induction in the case of the sequential composition is that each time we point into one part of the code, we make at least one step (in multistep), and thus in the other part, we have less steps left than we started with. This way we can decrease the value of $k$ and use the induction hypothesis for the single component and for a smaller $k$, respectively.

\subsection{Properties} 

In this subsection, we show that $\oplus$ is associative and commutative, and that the denotational semantics is compositional and prefix closed.

\subsubsection{Associativity and Commutativity}
\label{sssec:basic}

\begin{corollary}[Associativity and Commutativity]
$\oplus$ is associative and commutative, i.e.,

$\forall code_1, code_2, code_3$: 
\[
\llbracket code_1 \oplus code_2 \rrbracket = \llbracket code_2 \oplus code_1 \rrbracket
\]
\[
\llbracket (code_1 \oplus code_2) \oplus code_3 \rrbracket = \llbracket code_1 \oplus (code_2 \oplus code_3) \rrbracket
\]
\end{corollary}

\begin{proof}
This follows directly from the conformance. As $\cup_{sc}$ removes all structure, any structure on a given set of instructions has the same semantic function.
\end{proof}

\subsubsection{Compositionality} 
\label{ssec:compositionality}
A semantics is {\it compositional}, if the semantics of a combination of components is defined directly using the semantics of the components and not details of them.

\begin{theorem}[Compositionality]
  The denotational semantics is compositional, as the semantic function for sequential composition ({\sc D-seq, D-ext}) uses solely the semantic functions of the components and does not need further details about the components. \qed
\end{theorem}

\subsubsection{Prefix Closure}
\label{ssec:prefixclosure}
A set of traces is {\em prefix closed}, if for each trace in the set, all prefixes are contained in the set, too. We need to slightly adapt this definition, as the argument set of the denotation does not need to be prefix closed. So we formulate the property as an invariant over the semantic functions (see Section \ref{sec:calculus} for a definition of invariants):  

\begin{theorem}[Prefix Closure]
If the argument set $S$ is prefix closed, so is the returned set $\llbracket code\rrbracket(S)$.
\end{theorem}

\begin{proof}
The adapted prefix closed property holds: We only need to take a closer look at {\sc D-comm}, as it is the only rule changing the trace. Looking at {\sc D-comm}, we see that the yielded set is a union of the argument set $S$ and all newly calculated states.  As the new states only contain traces that are constructed by appending to existing traces, we get by using the assumption that all the prefixes of all constructed traces are also in the returned set.
\end{proof}

If we now assume that the set of initial states only contains empty traces, than the semantics for a given program returns sets which are prefix closed.

$\;$

We have shown important properties, which are a prerequisite to relate the denotational semantics and CSP and have outlined a proof that the denotational semantics relates the same states as the smallstep semantics. In the next section, we give a preliminary rule for an invariant-based Hoare calculus.

\section{Towards a Hoare calculus}
\label{sec:calculus}

Reasoning with semantics directly is generally complex. Therefore a proof calculus is usually introduced to make verification of properties manageable. We leave the proof calculus for future work, but we define an invariant and present a preliminary rule in this section.
 
In order to show properties of communicating, possibly nonterminating programs traditional Hoare logics with pre- and postconditions are not well suited, as the postcondition may never be reached. This is why we enhance the Hoare calculus with invariants to assert properties over the communication history. Using a calculus with invariants, we are still able to express properties about the program behavior even if the postcondition is never reached. This approach is, e.g., taken by Zwiers~\cite{zwiers1989compositionality}.
Invariants come quite naturally to reasoning about unstructured code, as the sequential composition is also the looping construct, and therefore already the sequential composition rule in the Hoare calculus defined by Saabas and Uustalu~\cite{Saabas05acompositional} uses invariants in pre- and postconditions.

We say an {\em invariant} $I$, holds for the semantic function of {\em code}, if it holds before, during, and after the execution for arbitrary initial states. Let $I$ be a predicate on a {\em state}, then we define formally:
\[
\inferrule*[right=Inv Intro]
{\forall S.\,(\forall s \in S.\,I(s)) \longrightarrow (\forall t \in \llbracket code\rrbracket(S).\, I(t) )}
{I \colon \llbracket code\rrbracket}
\]

In the example in the next section, we will use the following rule to reason about the components of sequential composition separately:
\[
\inferrule*[right=Inv $\oplus$]
{I \colon \llbracket code_1\rrbracket \\ I \colon\llbracket code_2\rrbracket}
{I \colon \llbracket code_1 \oplus code_2\rrbracket}
\]
We have proven soundness of the rule {\sc Inv $\oplus$} in Isabelle/HOL.
The intuition is that sequential composition leads to the repeated execution of the code sections $code_1$ and $code_2$ in a possibly alternating order, and if both keep the invariant valid, so does the sequential composition as it is compositional.  Admissibility of the invariant is ensured, as we defined it on single states.

\section{Example}
\label{sec:example}

\lstdefinelanguage{ucc}
{morekeywords={cbr, do, comm, br},
sensitive=false,
morecomment=[l]{//},
morecomment=[s]{/*}{*/},
morestring=[b]",
}

\begin{figure}
\lstset{
  language=ucc,
  tabsize=4,
}
\begin{lstlisting}[mathescape]
    1:: do $(\lambda \sigma.\; \{\sigma[free \gets \text{True}]\})$
$\oplus$    2:: comm $ef$ $f$ where
        $ef$ = $(\lambda \sigma. \;\, \{ in.x \; \,\,\mid \sigma(free) = \text{True} \wedge x \in T\}$
            $\qquad\cup \,\{out.x \mid \sigma(free) = \text{False} \wedge  x = \sigma(buffer) \} )$
        $f\,\,\,$  = $(\lambda \sigma\; event.\; \textbf{case}\; event \; \textbf{of}$
                $\mid in.x \;\,\,\Rightarrow \sigma[buffer \gets x, free \gets \text{False}]$
                $\mid out.x \Rightarrow \sigma[free \gets \text{True}])$
$\oplus$    3:: cbr $(\lambda \sigma. \; \text{True})$ 2 2

\end{lstlisting}

\caption{One Place Buffer in CUC}
\label{code1}
\end{figure}

We demonstrate the applicability of our formal framework as presented above and show for a simple implementation (written in CUC) of a one-place buffer that it is correct, i.e., it only outputs what was input before. We do so using the denotational semantics and a preliminary rule from the proof calculus. 
See Figure \ref{code1} for the code listing (we define $\oplus$ to be right associative). The code implements a buffer that can hold one value of type $T$. We explain the code line by line:\\
{\tt (1::do)} -- This is the initialization. The boolean {\it free} indicates that the {\it buffer} is ready to store data.\\
{\tt (2::comm)} -- The {\tt comm}-instruction both offers the events and changes the state after the communication happened. The events offered by $ef$ are all values of type T on channel $in$ if the buffer is free, else the output event with the value stored in the buffer is offered.
According to the event communicated, it either stores the input value and sets the buffer to not free, or it just sets the buffer to free.\\
{\tt (3::cbr)} -- The conditional branch is used in this case to model an unconditional branch and always jumps back to the {\tt comm}-instruction at label 2.\\

We show that, starting with an empty trace and with the program counter pointing to the first instruction (as described by the assertion $Pre$), only values that are input are also output (as described by $Inv$):

\[
Pre(tr, \sigma, pc) \coloneqq tr = \trace{} \wedge pc = 1
\]
\begin{align*}
&Inv(tr, \sigma, pc) \coloneqq tr \in TR_{even} \cup TR_{odd} \\
&\text{  where }
TR_{even} \coloneqq \{in.x \cat out.x \mid x \in T \}^{*{}}\\
&\text{  and } 
TR_{odd} \coloneqq \{tr\cat in.x \mid tr \in TR_{even} \wedge x \in T\}
\end{align*}

The operator $\cdot^{*{}}$ being the {\em Kleene Star}, $TR_{even}$ is the set of all {\em even} traces, i.e., where every input of  a value is followed by the output of the same value. $TR_{odd}$ is the set of {\em odd} traces, which are the even traces with an appended single input.

\noindent We name instructions by their label, e.g., in this example {\em instruction 3} means \mbox{(3:: cbr $(\lambda \sigma. \; \text{True})$ 2 2)}.

Following the structure of the code, we start reasoning about the sequential composition of instructions 2 and 3, and then combine it with instruction 1.

To show $Inv$ we need to strengthen it by adding state and program counter information, as the instructions depend on the state and on the program counter. This gives us a stronger invariant, denoted $I_{2,3}$, which we will use as invariant for instructions 2 and 3.

  \begin{align*}
        I_{2,3}(tr, \sigma, pc) \coloneqq &
        \;\;\, (tr \in TR_{even} \wedge \sigma(free) = \text{True}\\
        &\vee tr \in TR_{odd} \wedge \sigma(free) = \text{False} \\
        &\qquad \wedge \sigma(buffer) = x
        \wedge
         \exists tr'. \,tr = tr' \cat in.x) \\
         & \wedge pc \in \{2, 3\}
  \end{align*}
We use the rule {\sc Inv $\oplus$} to reason about instructions 2 and 3 separately.

\noindent $I_{2,3}\colon \llbracket2:: \texttt{comm}\; ef\; f\rrbracket$ -- Either we start with an even trace ($\in TR$) and $free = \text{True}$, then we accept input and the second disjunct holds or we start in the second disjunct with an odd trace, output the corresponding value, and then end up in the first disjunct. In both cases we end up with $pc = 3$.
 
\noindent $I_{2,3}\colon \llbracket3:: \texttt{cbr} \;(\lambda \sigma. \; \text{True})\;\; 2 \;\;2\rrbracket$ -- \texttt{cbr} does neither change state nor trace, only the program counter is set to 2, so the validity of $I_{2,3}$ is preserved by instruction 3.

We are left to show that instruction 1 brings us from the initial state to the $I_{2,3}$:
Assume $Pre$ holds, then all states reachable through instruction 1 either still fulfill $Pre$ or 
\[
I_2(tr, \sigma, pc) \coloneqq
tr = \trace{} \wedge \sigma(free) = \text{True} \wedge pc = 2
\]

where the latter implies $I_{2,3}$. As instruction 1 does not change the state when $pc \in \{2,3\}$ holds,

\[
I_{1,2,3}(tr, \sigma, pc) \coloneqq Pre(tr, \sigma, pc) \vee I_{2,3}(tr, \sigma, pc)
\]
is an invariant for instruction 1.

Following a similar argumentation, $I_{1,2,3}$ is also an invariant for instructions 2 and 3, as neither does change the state where $pc=1$ holds and neither sets $pc=1$, so in the first case the validity of $Pre$ is maintained, and in the latter case $Pre$ as a disjunct (of the postcondition) can be ignored. 
The easy combination of disjuncts is due to the fact that all instruction labels are unique. The future proof calculus will make use of this property to combine invariants.

As $I_{1,2,3}$ holds for both instruction 1 and instructions 2 and 3, we deduce using {\sc Inv $\oplus$} that $I_{1,2,3}$ holds for all three instructions.
As $Inv$ follows from $I_{1,2,3}$, we have shown that starting in $Pre$ the validity of $Inv$ is maintained by the program.

We have shown how a preliminary Hoare calculus build on the denotational semantics can be used to reason about a nonterminating program. The proved property is an assertion about traces. Being able to show properties about traces will help to relate CUC and CSP.

\section{Conclusion}
\label{sec:conclusion}

In this paper, we have defined a compositional denotational semantics for a generic low-level language. Our denotational semantics is based on least fixpoints, it copes with communication, and it enables reasoning about nonterminating programs. We have proved that it conforms to the more intuitive operational semantics of our language. We have analyzed our semantics with respect to compositionality and prefix closure and have demonstrated how it could be used to verify properties of unstructured code using a proof rule based on invariants.

In future work, we aim at using our denotational semantics as the basis for a refinement proof calculus that enables convenient conformance proofs between unstructured low-level programs and CSP-based specifications. To this end, we have based the behavioral semantics part within our denotational semantics on denotations as used in the CSP process calculus. 
To realize such a refinement proof calculus, we will first define a Hoare calculus for our semantics from which we have shown a preliminary rule. The refinement proof calculus will relate CSP processes and proof obligations in the Hoare calculus.
To fully support CSP, we will extend the denotational semantics with concurrency in the form of an {\em alphabetized parallel operator} as in CSP, i.e., synchronous progress on shared events, otherwise interleaving. To not only be able to show safety properties but also liveness properties using such a refinement proof calculus, we want to extend the denotational semantics and the Hoare calculus to cope with failures.

\newcommand{\doi}[1]{\textsc{doi}: \href{http://dx.doi.org/#1}{\nolinkurl{#1}}}
\bibliographystyle{eptcs}
\bibliography{../bibliography}

\end{document}